\documentclass[11pt,letterpaper]{article}

\usepackage[margin=1in]{geometry}
\usepackage{theorem,latexsym,graphicx,amssymb}
\usepackage{amsmath,enumerate}
\usepackage{bbm}
\usepackage{wrapfig}
\usepackage{subfigure}
\usepackage{float}
\usepackage{psfig}
\usepackage{epsfig}
\usepackage{xspace}
\usepackage{paralist}
\usepackage{enumerate}
\usepackage{cases}
\usepackage{caption}
\usepackage{algorithm}
\usepackage[noend]{algpseudocode}
\usepackage{color}
\usepackage{complexity}
\usepackage{multicol}
\usepackage{thm-restate}

\newenvironment{proof}{{\bf Proof:  }}{\hfill\rule{2mm}{2mm}\vspace*{5pt}}

\numberwithin{figure}{section}
\numberwithin{equation}{section}

\newtheorem{corollary}{Corollary}[section]
\newtheorem{theorem}{Theorem}[section]
\newtheorem{lemma}{Lemma}[section]

\newcommand{\vecy}{\vec{y}}
\newcommand{\expect}[2]{\underset{#1}{\operatorname{\mathbf E}}\left[#2\right]}

\title{A Simple $1-1/e$ Approximation for Oblivious Bipartite Matching}

\author{Zhihao Gavin Tang\thanks{ITCS, Shanghai University of Finance and Economics. {\texttt{tang.zhihao@mail.shufe.edu.cn}}}
	\and Xiaowei Wu\thanks{Faculty of Computer Science, University of Vienna. {\texttt{wxw0711@gmail.com}}}
	\and Yuhao Zhang\thanks{Department of Computer Science, The University of Hong Kong. {\texttt{yhzhang2@cs.hku.hk}}}}
\date{\today}

\begin{document}

\begin{titlepage}
	\thispagestyle{empty}
	\maketitle
	
	\begin{abstract}
		We study the oblivious matching problem, which aims at finding a maximum matching on a graph with unknown edge set.
		Any algorithm for the problem specifies an ordering of the vertex pairs.
		The matching is then produced by probing the pairs following the ordering, and including a pair if both of them are unmatched and there exists an edge between them.
		The unweighted (Chan et al. (SICOMP 2018)) and the vertex-weighted (Chan et al. (TALG 2018)) versions of the problem are well studied.
		
		In this paper, we consider the edge-weighted oblivious matching problem on bipartite graphs, which generalizes the stochastic bipartite matching problem.
		Very recently, Gamlath et al. (SODA 2019) studied the stochastic bipartite matching problem, and proposed an $(1-1/e)$-approximate algorithm.
		We give a very simple algorithm adapted from the \textsf{Ranking} algorithm by Karp et al. (STOC 1990), and show that it achieves the same $1-1/e$ approximation ratio for the oblivious matching problem on bipartite graph.
	\end{abstract}
\end{titlepage}

\section{Introduction}

Motivated by efficient maximal matching computation and the kidney exchange applications~\cite{jet/RothSU05}, the oblivious matching problem (defined as follows) has drawn lots of attention in recent years.

\paragraph{Oblivious Matching.}
The adversary fixes a graph $G=(V,E)$, and only reveal the set of vertices to the algorithm.
That is, the algorithm has no information on the edge set.
At each step, the algorithm probes a pair of unmatched vertices $(u,v)$, and includes the pair in the matching irrevocably if there exists an edge between them, a.k.a., the query-commit model.

\medskip

Any (randomized) algorithm determines a sequence of vertex pairs, with the objective of maximizing the (expected) size of matching produced. 

It is easy to show that any algorithm that produces a maximal matching is $0.5$-approximate, and this is the best approximation ratio of any deterministic algorithm.
The first randomized algorithm beating the $0.5$ barrier was achieved by Aronson et al.\cite{rsa/AronsonDFS1995}, who showed that the Modified Randomized Greedy (MRG) algorithm is $(1/2+1/400000)$-approximate.
Better approximation ratios for the problem have also been obtained~\cite{sicomp/ChanCWZ18,talg/ChanCW18}, using the \textsf{Ranking} algorithm proposed by Karp et al.~\cite{stoc/KarpVV90} for the online bipartite matching problem.
For the problem on bipartite graphs, Mahdian and Yan~\cite{stoc/MahdianY11} showed that the \textsf{Ranking} algorithm achieves an approximation ratio $0.696$, strictly larger than $1-\frac{1}{e}$.

In this paper we consider the edge-weighted version of the problem on bipartite graphs.
In the edge-weighted setting, there is a weight $w_{uv}$ associated with each pair of vertices, which is the weight of the edge $(u,v)$, if such an edge exists.

\paragraph{Prior Works.} 
The vertex-weighted setting\footnote{In the vertex-weighted version, there is a weight $w_u$ associated with each vertex $u\in V$, and $w_{uv} = w_u + w_v$.} of the problem is studied by Chan et al.~\cite{talg/ChanCW18}. They observed that the weighted version of \textsf{Ranking} by Aggarwal et al.~\cite{soda/AggarwalGKM11} achieves an $1-\frac{1}{e}$ approximation ratio for the problem on bipartite graphs.
They also proved that the same algorithm achieves an approximation ratio strictly larger than $0.5$ on general graphs.
Very recently, Tang et al.~\cite{corr/TangWZ19} proposed an algorithm for the edge-weighted oblivious matching problem on general graphs that is $0.501$-approximate, which is the first to beat the $0.5$ approximation ratio by Greedy.

The stochastic matching problem~\cite{soda/GamlathKS19,soda/Singla18,sigecom/AssadiKL17,icalp/CostelloTT12} can be regarded as an ``easier'' version of the oblivious matching problem.
In the stochastic setting, in addition to the weight $w_{uv}$ associated with each pair $(u,v)$, there is a probability $p_{uv}$.
When a pair $(u,v)$ is probed, the edge exists with probability $p_{uv}$, and the existences of all edges are independent random events.
Gamlath et al.~\cite{soda/GamlathKS19} considered the problem on bipartite graphs and proposed an $(1-\frac{1}{e})$-approximation algorithm.


\subsection{Our Results}
Our main contribution is a proper generalization of the \textsf{Ranking} algorithm to edge-weighted graphs.
Our algorithm is consistent with the algorithm on unweighted graphs~\cite{stoc/KarpVV90} and on vertex-weighted graphs~\cite{soda/AggarwalGKM11}.
We show that our algorithm is $(1-1/e)$-approximate for the edge-weighted oblivious matching problem on bipartite graph.
Surprisingly, our analysis is very simple, and is a straightforward adaption of the randomized primal-dual proof of \cite{soda/DevanurJK13}.

\begin{theorem}
	\label{thm:main}
	The weighted \textsf{Ranking} algorithm is $(1-1/e)$-approximate for the edge-weighted oblivious matching problem on bipartite graph.
\end{theorem}

Obviously, any algorithm for the oblivious matching problem applies to the stochastic setting by ignoring the additional probability information.
Indeed, since we do not need this extra information, the same approximation ratio can be achieved even if the probabilities are arbitrarily correlated.


\begin{corollary}
	The weighted \textsf{Ranking} algorithm is $(1-1/e)$-approximate for the stochastic bipartite matching problem.
\end{corollary}

We remark that our algorithm achieves the same approximation ratio $(1-1/e)$ as the algorithm by Gamlath et al.~\cite{soda/GamlathKS19}, while the two algorithms exploit quite different structures of the problem.
Indeed, the existence probabilities of edges are crucial to \cite{soda/GamlathKS19} in that they can estimate the probability of each edge appearing in the optimal matching in advance. Given that the $(1-1/e)$ analysis of both algorithms are tight, it remains an interesting open question to see how the two ideas can be combined and how better algorithms can be designed for the stochastic bipartite matching problem.

\section{\textsf{Ranking} on Edge-weighted Bipartite Graphs}
Let the given bipartite graph be $G=(L\cup R,E)$, where $L$ and $R$ denote the left hand side and right hand side vertex set, respectively. Next, we describe the weighted \textsf{Ranking} algorithm:

\paragraph{Weighted \textsf{Ranking} Algorithm.}
Fix non-decreasing function $g(x):=e^{x-1} $.  Each vertex $u\in L$ independently draws a rank $y_u\in[0,1]$ uniformly at random. For each pair of vertices $(u,v)$ where $u \in L$ and $v \in R$, let $(1-g(y_u)) w_{uv}$ be the perturbed weight of pair $(u,v)$. We probe all pairs of vertices in descending order of their perturbed weights. 

\medskip

Before the analysis, we introduce some notations. Let $\vec{y}$ denote the rank vector of all vertices in $L$ and $M(\vec{y})$ denote the corresponding matching produced by our algorithm with the rank vector $\vec{y}$.
Let $W^*$ be weight of a maximum weight matching of the given graph $G$.
A (randomized) algorithm is $r$-approximate, where $r\in[0,1]$, if for any given graph, the (expected) weight of matching given by the algorithm is at least $r \cdot W^*$.

As mentioned in the introduction, our analysis is built on a randomized primal-dual framework, which is first introduced by \cite{soda/DevanurJK13}.

For each edge matched by our algorithm, we set the dual variables of the two endpoints so that the summation equals the edge weight.
By doing so, the summation of dual variables (which is the dual objective) equals the weight of the matching.
Since the algorithm is randomized, the dual variables are random variables depending on the ranks of vertices.
However, as long as we can show that in expectation the dual constraints are approximately feasible, then we can give a lower bound on the approximation ratio using the approximate feasibility.

Formally, we import the following lemma from~\cite[Lemma 2.1]{corr/TangWZ19} and~\cite[Lemma 2.6]{stoc/HKTWZZ18} that extend the randomized primal-dual framework of~\cite{soda/DevanurJK13}.

\begin{lemma}\label{lemma:dual}
	If there exist non-negative random variables $\{ \alpha_u \}_{u\in V}$ depending on $\vecy$ such that
	\begin{compactitem}
		\item[(1)] for every rank vector $\vecy$ of vertices, $\sum_{u\in V} \alpha_u = \sum_{(u,v)\in M(\vecy)}w_{uv}$;
		\item[(2)] for every $(u,v)$ matched in the maximum weight matching, $\expect{\vecy}{\alpha_u + \alpha_v} \geq r\cdot w_{uv}$,
	\end{compactitem}
	then our algorithm is $r$-approximate.
\end{lemma} 


\medskip

Next, we define the specific gain sharing method, i.e. how the dual variables $\alpha$ are chosen.

\paragraph{Gain Sharing.}
For each matched edge $(u,v)$, where $u\in L$ and $v\in R$, let $\alpha_u = g(y_u)\cdot w_{uv}$ and $\alpha_v = (1-g(y_u)) \cdot w_{uv}$. We shall refer to $\alpha_u, \alpha_v$ as the gains of vertices $u,v$.

\medskip

It is obvious to see the first condition of Lemma~\ref{lemma:dual} is satisfied by our gain sharing rule, since for any matched edge $(u,v)$, we have $\alpha_u + \alpha_v = w_{uv}$.

Next, we fix a pair of neighbors $(u,v)$ and derive a lower bound of the expect gain of $\alpha_u + \alpha_v$, which is used to satisfy the second condition of Lemma~\ref{lemma:dual}.
We start from a basic monotone property of the ranks of vertices, e.g., higher rank leads to a ``better'' matching for a fixed vertex.

\begin{lemma}[Monotonicity]\label{lemma:monotonicity}
	Consider any matching $M(\vec{y})$ and any vertex $u \in L$, if we fix the ranks of all vertices expect $u$, the weight of edge $u$ matches is non-increasing w.r.t. $y_u\in [0,1]$.
\end{lemma}
\begin{proof}
	Suppose $u$ is matched with $v$ when $y_u = y$.
	If we fix the ranks of all other vertices, and decrease $y_u$ to be some $y'<y$, the perturbed weight of all edges adjacent to $u$ will decrease.
	Hence these edges have a more prior order when we probe edges.
	More specifically, edge $(u,v)$ will be probed at least as early as when $y_u = y$.
	If $u$ is unmatched when $(u,v)$ is probed, then $u$ and $v$ will match each other and the matching remains the same; otherwise $u$ is matched to some $z$ such that
	\begin{equation*}
		(1-g(y_u)) \cdot w_{uz} \geq (1-g(y_u))\cdot w_{uv},
	\end{equation*}
	which implies $w_{uz}\geq w_{uv}$. Hence the lemma follows.
\end{proof}

From now on, we fix the ranks of all vertices in $L$ other than $u$. By Lemma~\ref{lemma:monotonicity}, there exists a \emph{marginal rank $\theta$}, such that the weight of edge $u$ matches is at least $w_{uv}$ if and only if $y_u\in[0,\theta)$. This implies a basic bound for the gain of $u$ when $y_u \in [0,\theta)$, and the gain of $v$ when $y_u \in [\theta,1]$.

\begin{lemma}
	\label{lemma:basic-gain}
	When $y_u \in [0,\theta)$, $\alpha_u \geq g(y_u) \cdot w_{uv}$, and when $y_u \in [\theta,1]$, $\alpha_v \geq (1-g(\theta))\cdot w_{uv}$.
\end{lemma}
\begin{proof}
	The first bound is implied from the definition of $\theta$ and our gain sharing method.
	For the second one, consider when $y_u = \theta$.
	By definition, $u$ matches an edge with weight smaller than $w_{uv}$, which means that when we probe edge $(u,v)$, which has perturbed weight $(1-g(\theta))\cdot w_{uv}$, $v$ is already matched and $u$ is not matched.
	Thus, the perturbed weight of the edge $v$ matches is larger than $(1-g(\theta))\cdot w_{uv}$, which implies $\alpha_v \geq (1-g(\theta))\cdot w_{uv}$.
	
	Since $u$ is not matched when we probe edge $(u,v)$, if we further increase $y_u$, then $u$ remains unmatched when we probe edge $(u,v)$ (by Lemma~\ref{lemma:monotonicity}).
	In other words, all edges adjacent to $u$ that are probed before $(u,v)$ are unsuccessful.
	Hence increasing $y_u$ to any value in $(\theta,1]$ does not change the matching status of $v$ ,which implies $\alpha_v \geq (1-g(\theta))\cdot w_{uv}$ for all $y_u\in[\theta, 1]$.
\end{proof}

The next lemma characterizes the matching status of $v$ when $y_u < \theta$ and is the only part of the proof that crucially uses the bipartiteness of the graph. 
Indeed, the lemma is implied by~\cite[Corollary 2.3]{corr/TangWZ19} by fixing the ranks of all $v\in R$ to be $1$ in their algorithm.

For completeness we present a sketch of the proof here.

\begin{lemma}[Extra Gain for Bipartite]
	\label{lemma:extra-gain}
	When $y_u \in [0,\theta)$, $\alpha_v \geq (1-g(\theta))\cdot w_{uv}$.
\end{lemma}
\begin{proof}
	We fix the ranks of all vertices in $L$ other than $u$.
	By definition when $y_u = \theta$, $v$ is matched and $u$ is unmatched when edge $(u,v)$ is probed. 
	Suppose $v$ is matched with $z\in L$, where $(1-g(y_z))\cdot w_{zv} \geq (1-g(\theta))\cdot w_{uv}$.
	Let $\vec{y}_1$ be the rank vector when $y_u = \theta$.
	Let $M_1$ be the partial matching right after we probe $(z,v)$ when $\vec{y} = \vec{y}_1$.
	Note that $u$ is not matched in $M_1$.
	Now consider the case when $y_u = y\in [ 0,\theta )$ and let $\vec{y}_2$ be the corresponding rank vector.
	Note that $\vec{y}_1$ and $\vec{y}_2$ differ in only $y_u$.
	Let $M_2$ be the partial matching right after we probe $(z,v)$ when $\vec{y}=\vec{y}_2$.
	
	If $u$ is unmatched in $M_2$ then $M_1$ and $M_2$ are identical because both matchings can be produced by removing $u$ from the graph and probing the edges up to $(z,v)$.
	Otherwise, suppose $u$ matches some $v_1\in R$ (it is possible that $v_1 = v$) in $M_2$.
	
	It is easy to show\footnote{For a formal proof, please refer to~\cite[Lemma 2]{sicomp/ChanCWZ18}, or~\cite[Lemma 2.3]{stoc/HKTWZZ18}, or~\cite[Lemma 2.3]{corr/TangWZ19}.} (by induction on the number of pairs probed) that the symmetric difference between $M_1$ and $M_2$ is an alternating path $(u = u_0,v_1,u_1,v_2,u_2,\ldots)$ starting from $u$ such that
	\begin{itemize}
		\item for all $i = 1,2,\ldots$, $(u_{i-1} v_{i}) \in M_2$; $(v_{i},u_i) \in M_1$.
		\item  for all $i=1,2,\ldots$, $w_{u_{i-1}v_i} \geq w_{v_i u_i}$.
	\end{itemize}
	
	If $v$ is not contained in the alternating path then $v$ is matched to the same vertex $z$ in both $M_1$ and $M_2$.
	Otherwise suppose $v=v_i$.
	Hence we have $u_i = z$ and $w_{u_{i-1} v} \geq w_{vz}$.
	
	In other words, $v$ is matched to some $u_{i-1}$ in $M_2$ such that $w_{u_{i-1} v}$ is at least $w_{vz}$.
	Hence for all $y_u = y\in [ 0,\theta )$, we have $\alpha_v \geq (1-g(\theta))\cdot w_{uv}$.
	%
	%
	%
	%
\end{proof}

Finally, we finish the proof of Theorem~\ref{thm:main}. Combining Lemma~\ref{lemma:basic-gain} and Lemma~\ref{lemma:extra-gain}, we have
\begin{align*}
	\expect{\vecy}{\alpha_u + \alpha_v} \geq w_{uv} \cdot \left(\int_0^\theta g(y_u)dy_u + (1-g(\theta))\right)
	= w_{uv} \cdot (1-\frac{1}{e}),
\end{align*}
where the equality follows from $g(y) = e^{y-1}$.
By Lemma~\ref{lemma:dual}, the approximation ratio follows.

{
		\bibliography{matching}
		\bibliographystyle{plain}
}
\end{document}